\documentclass[letterpaper, 10 pt, conference]{ieeeconf}  

\IEEEoverridecommandlockouts                              

\title{\LARGE \bf
Detecting Feedback-path Delay Injection Attacks Using Interacting Multiple Model Filtering}

\author{Lovisa Eriksson, Torbjörn Wigren, Dave Zachariah and André M. H. Teixeira
\thanks{This work is supported by the Swedish Research Council under grants 2021-05022, 2023-05234, and 2024-03903, the Swedish Foundation for Strategic Research, and the Knut and Alice Wallenberg Foundation. The authors are with Department of Information Technology, Uppsala University, Uppsala, Sweden. \texttt{\{lovisa.eriksson, torbjorn.wigren,
dave.zachariah, andre.teixeira\}@it.uu.se}}
}

\usepackage[textsize=tiny]{todonotes}

\usepackage{color}

\usepackage{amsfonts}
\usepackage{amsmath}
\usepackage{amssymb}
\usepackage{algorithm}
\usepackage{graphicx}
\usepackage{amsfonts}
\usepackage{amsmath}
\usepackage{amssymb}
\usepackage{multicol}
\usepackage{multirow}
\usepackage{subcaption}
\usepackage{adjustbox}

\usepackage{tikz}
\usetikzlibrary{positioning}

\usepackage{algorithm}
\usepackage{algpseudocode}

\usepackage{tikz}
\usepackage{pgfplots}
\pgfplotsset{compat=1.14}
\usepgfplotslibrary{external}
\usetikzlibrary{arrows.meta,calc,chains,shapes.geometric}

\newcommand{\tcr}[1]{{\color{black} #1}}

\newcommand{\expect}{\mathbb{E}}
\newcommand{\prob}{\mathbb{P}}

\newcommand{\pfa}{\textsc{pfa}}
\newcommand{\add}{\textsc{add}}
\newcommand{\arl}{\textsc{arl}}

\newcommand{\detector}{\textsc{alarm}}
\newcommand{\statprob}{\overline{\pi}}

\newcommand{\maxdelay}{N}
\newcommand{\etime}{\tilde{t}}
\newcommand{\attacktime}{t_*}

\newtheorem{remark}{Remark}

\newtheorem{assumption}{Assumption}

\newtheorem{definition}{Definition}

\newtheorem{proposition}{Proposition}

\begin{document}

\maketitle
\thispagestyle{empty}
\pagestyle{empty}

\begin{abstract}
\label{abstract}
Time-delays are known to have a detrimental effect on feedback systems. In the context of networked cyber-physical systems, delays can be injected by malicious adversaries. Detecting them early is an important challenge. This paper proposes a novel variation of Interacting Multiple Model filtering to detect delay injection attacks in feedback control systems, when hidden in an open loop setting. The detection scheme is formalised by treating delay as alternative modes of the system, and theoretical analysis of the stationary distribution informs a reduction to a three parameter model as well as the choices of hyper parameter values. The method is evaluated on a cruise control application, and shows detection within a few seconds and a low false alarm probability.

\end{abstract}

\section{INTRODUCTION}
\label{introduction}
Time delays can cause serious issues in feedback control systems \cite{fridman2014}, particularly in control over networks~\cite{hespanha2007}, leading to poor performance, lack of robustness, and possibly loss of stability. Therefore, delay injection attacks have been studied in multiple previous works \cite{wigren2023, abbasspour2019, bianchin2018, korkmaz2017, xin2020}. These suggest various mitigation techniques, such as hypothesis testing \cite{korkmaz2017} and Neural Network delay estimators \cite{abbasspour2019, xin2020}.

This paper follows \cite{wigren2023} and considers delay attacks where an unknown delay is inserted in the feedback loop, so that it only affects the observations, as illustrated in Fig. \ref{fig:block_diagram}. Both the amount of delay and time of attack are unknown, but the system dynamics are assumed to be known. By adapting the widely applied Interacting Multiple Model (IMM) filter \cite{bar2001, bar2005} to the problem of delay detection, a set of Kalman filters for various modes (i.e., delay hypotheses) is constructed, and propagated at each time step. The estimates from all filters at the previous time step 
are mixed according to the probabilities of the modes.

By determining how well the estimates of each mode fit the observations, the probability of each hypothesis is calculated in a Bayesian fashion. When the probability of `no attack' falls below a threshold, an attack is detected, and action can be taken to prevent reduced performance or destabilisation of the feedback control loop.

A similar approach for delay detection was proposed in \cite{korkmaz2017}, where hypothesis testing is applied to check how well the data fit different assumptions on delay. This differs from the approach in the present paper by not allowing mixing of estimates; instead it is assumed at each time step that there was no delay in the previous time step and no Markov transitions are used. Further, \cite{korkmaz2017} does not specifically consider delays in the feedback path. In \cite{wigren2023}, on the other hand, the attack is in the feedback path but the delay is estimated together with the full system dynamics, making it more flexible but less accurate than the proposed IMM based method. None of the previous papers develops a formal detection scheme, but rather plots estimated delays so that a human overseer could make manual decisions on whether actions are necessary. The method proposed in the current paper allows automatic detection, and is faster than both above mentioned schemes.

\begin{figure}
    \centering
    \tikzset{
block/.style = {align=center, draw, fill=white, rectangle, minimum height=3em, minimum width=3em, align=center},
redblock/.style = {align=center, draw=red, fill=white, rectangle, minimum height=3em, minimum width=3em, align=center},
blueblock/.style = {align=center, draw=blue, fill=white, rectangle, minimum height=3em, minimum width=3em, align=center},
greyblock/.style = {align=center, draw=gray, fill=white, rectangle, minimum height=3em, minimum width=3em, align=center},
tmp/.style  = {coordinate}, 
sum/.style= {draw, fill=white, circle, node distance=1cm},
input/.style = {coordinate},
output/.style= {coordinate},
pinstyle/.style = {pin edge={to-,thin,black}
}
}
\scalebox{0.65}{
\begin{tikzpicture}[auto, node distance=2cm,->]
    
    \node [input, name=control_input] (control_input) {};
    \node [input, name=manual_input, above of=control_input] (manual_input) {};
     \node [input, name=switchL, right of=control_input] (switchL) {};
     \node [input, name=switchL_above, above of=switchL] (switchL_above) {};
     \node [input, name=switchR, right of=switchL, node distance=1cm] (switchR) {};
    \node [block, name=system, right of=switchR, node distance = 3cm] (system) {State propagation \\ $x_{t+1} = Ax_t + Bu_t + w_t$};
    \node[output, name=output1, right of=system, node distance = 3cm] (output1){};
    \node[output, name=output2, right of=output1] (output2){};
     \node [redblock, name=observation, below of=output1, node distance = 2cm] (observation) {\textcolor{red}{(Delayed) }Observation  $y_{t+1} =$\\ $Cx_{t+1-\textcolor{red}{\delta_t}} + \nu_t$};
     \node[output, name=obs_est, below of=system, node distance = 4cm](obs_est){};
    \node [greyblock, name=state_estimate, left of=obs_est, node distance = 3cm] (state_estimate) {State estimator};
    \node [greyblock, name=controller, below of=control_input, node distance = 2cm] (controller) {Controller};
    \node[input, name=lowerleft, below of=controller](lowerleft){};
    \node[blueblock, name=detector, below of=obs_est](detector){Detector};
    \node[input, name=x_t, left of=detector, node distance = 3cm](x_t){};
    \node[input, name=lowerright, below of=observation](lowerright){};

    \draw [->] (switchR) -- (system);
    \draw [-] (control_input) -- node{$u_t$}(switchL);
    \draw [-] (manual_input) -- node{$u_t$}(switchL_above);
    \draw [-] (switchL_above) -- node{open/closed} (switchR);
    \draw [-] (system) -- node{$x_{t+1}$}(output1);
    \draw[->] (output1) -- (output2);
    \draw[->] (output1) -- (observation);
    \draw[-] (observation) -- (lowerright);
    \draw[-] (lowerright) -- (obs_est);
    \draw[->] (obs_est) -- node{$y_{t+1}$} (state_estimate);
    \draw[-] (state_estimate) -- node{$\hat{x}_{t+1}$}(lowerleft);
    \draw[->] (lowerleft) -- (controller);
    \draw[-] (controller) -- (control_input);
    \draw[->] (obs_est) -- (detector);
    
\end{tikzpicture}}
    \caption{Block diagram of the model. The attack only affects the observation.}
    \label{fig:block_diagram}
\end{figure}

Inserting the attack only in the feedback path ensures it goes unnoticed in an open loop setting, but still causes problems when the loop is closed. For a more thorough motivation of the attack placement, see \cite{wigren2023}.

The contributions of this paper are as follows:
\begin{enumerate}
\item The IMM-based scheme for delay attack detection is formulated under a quickest change detection setting, leveraging the posterior probabilities as the test statistic for the detection rule;
\item Based on the underlying Hidden Markov Model (HMM), the theoretical analysis provides guidelines for tuning the hyperparameters of the HMM and the detection threshold;
\item The numerical example on a safety-critical application, cruise control, illustrates the applicability of the method and shows good detection performance.
\end{enumerate}

The paper is organised as follows. Section \ref{state_space_model} presents the formal problem, specifying the system model and the delay attacks. The evaluation metrics to assess the performance of the detection scheme are defined in Section \ref{evaluation_metrics}. Section \ref{interacting_multiple_models}  describes how IMM can be applied for delay detection in a general setting, Section \ref{construction_of_the_markov_model} how the HMM is designed, and Section \ref{stationary_state_analysis} how the remaining parameters are set. In Section \ref{delay_detection_in_cruise_control}, the model is applied to cruise control, and simulation results are presented and discussed. The paper is concluded in Section \ref{conclusions} with final remarks and suggestions for future work.

\section{PROBLEM FORMULATION}
\label{problem_formulation}

\subsection{State Space and Attack Model}
\label{state_space_model}
Consider the linear time-invariant system model
\begin{equation}
\label{eq:state_space_model}
\begin{cases}
    x_{t+1} &= Ax_t + Bu_t + w_t \\
    y_t &= Cx_{t-\delta_t}  + \nu_{t-\delta_t}, 
    \end{cases}
\end{equation}
with independent white noise processes $w_t \sim \mathcal{N}(0, Q)$,  $\nu_t \sim \mathcal{N}(0, R)$. The system model parameters $A, B, C$, as well as $R$ and $Q$, are known. However, the \emph{observed} output $y_t$ is subject to a (possibly) time-varying delay $\delta_t$. For ease of evaluation, the system is assumed to have no inherent delay. 
A delay injection attack adds an arbitrary delay into the observations.
No further assumptions are made, and so the delay size can depend arbitrarily on time. This is formalised as Assumption \ref{ass:delay}.

\begin{assumption}
\label{ass:delay}
    The delay $\delta_t$ is a function $\delta:\mathbb{N}\rightarrow\mathbb{N}$, such that $\delta_t = 0$ for $t<\attacktime$ and $\delta_t>0$ for $t\geq \attacktime$, where $\attacktime$ denotes an \emph{unknown} attack time. Otherwise, the function class of the delay as well as its current value are \emph{unknown}.
\end{assumption}

This assumption is weak and allows a wide range of attack schemes, e.g. constant delay, randomized delay, slowly increasing delay, and replay attacks. 

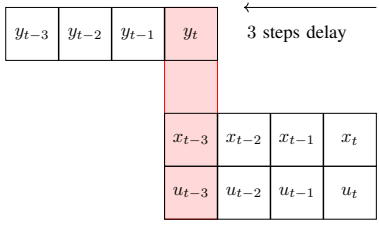
\begin{figure}
    \centering
    \scalebox{0.7}{\begin{tikzpicture}[baseline=(current bounding box.north)]
\draw[red, fill=red!15] (3,-3) rectangle (4,1);
    \draw[black] (0,0) rectangle (4,1);
\node[align=right] at (0.5,0.5) {$y_{t-3}$};
\draw[black, thick] (1, 0) -- (1,1);
\node[align=right] at (1.5,0.5) {$y_{t-2}$};
\draw[black, thick] (2, 0) -- (2,1);
\node[align=right] at (2.5,0.5) {$y_{t-1}$};
\draw[black, thick] (3, 0) -- (3,1);
\node[align=right] at (3.5,0.5) {$y_{t}$};

\draw[black] (3,-2) rectangle (7,-1);

\node[align=right] at (3.5,-1.5) {$x_{t-3}$};
\draw[black, thick] (4, -2) -- (4,-1);
\node[align=right] at (4.5,-1.5) {$x_{t-2}$};
\draw[black, thick] (5, -2) -- (5,-1);
\node[align=right] at (5.5,-1.5) {$x_{t-1}$};
\draw[black, thick] (6, -2) -- (6,-1);
\node[align=right] at (6.5,-1.5) {$x_{t}$};
\node[align=left] at (5.5, 0.5){$3$ steps delay};
\draw[<-] (4.5, 1) -- (7, 1);

\draw[black] (3,-3) rectangle (7,-2);
\node[align=right] at (3.5,-2.5) {$u_{t-3}$};
\draw[black, thick] (4, -3) -- (4,-2);
\node[align=right] at (4.5,-2.5) {$u_{t-2}$};
\draw[black, thick] (5, -3) -- (5,-2);
\node[align=right] at (5.5,-2.5) {$u_{t-1}$};
\draw[black, thick] (6, -3) -- (6,-2);
\node[align=right] at (6.5,-2.5) {$u_{t}$};
\end{tikzpicture}}
    \caption{Illustration of observations delayed by $3$ time steps. The observed output at time $t$ reflects the actual state at time $t-3$.}
    \label{fig:example_delay}
\end{figure}

An example of the resulting observations for $3$ time steps of delay is shown in Fig. \ref{fig:example_delay}, where it can be seen how the observation at time $t$ is not a function of the true state at that time, but rather the true state at time $t-3$. 

This attack model, illustrated in Fig. \ref{fig:block_diagram}, is parsimonious, in that it only affects the observations via the feedback path, yet effective: it does not affect the system's states in open loop, and thus remains hidden, but is often able to cause destabilisation of the closed-loop system. For a more thorough motivation of placing attacks in the feedback path, see \cite{wigren2023}. 

\begin{remark}It should be noted that the delay is only inserted in the feedback path. As detection is performed in the open loop setting, the state estimator and controller in Fig. \ref{fig:block_diagram}, do not affect the detection performance, but is still included to highlight the potential destabilizing effect of the attack if the loop is closed.\end{remark}

The potential to switch between open and closed loop settings is common in semi-automated tasks, where a human can switch between determining the input $u_t$ and using a feedback controller. Two safety-critical examples are cruise control and autopilot systems in cars \cite{nilsson2016, yueming2011} and aircraft \cite{bryson1975}. 

\subsection{Attack Detector and its Performance}
\label{evaluation_metrics}
Attack detection is performed by a detector $\detector(t)$, detecting the presence of a delay attack, i.e. $\delta_t>0$, in open loop. The detector is formally defined in Definition \ref{def:detector}.
\begin{definition}
\label{def:detector}
    A delay detector is a function
    \begin{equation}
        \detector(t): U^t\times Y^t \rightarrow \{0,1\},
    \end{equation}
    where $U$ is the input domain and $Y$ the measurement domain.
    The detector is
    \begin{align}
        \detector(t)&\equiv \begin{cases}0,&\pi((u_{k-1}, y_k)_{k=1}^t) >\xi \\
1, &\pi((u_{k-1}, y_k)_{k=1}^t) \leq \xi \\
\end{cases}
    \end{align}
    where $\pi$ is a function of the data up to time $t$, $\xi$ a threshold,  and $\detector(t)=1$ denotes that a delay has been detected.
\end{definition}
With this formalisation of the detection scheme, two quantities are needed: the function $\pi$ and the threshold $\xi$. This paper focuses mainly on how to set $\pi$ through IMM estimation, but a short discussion on setting the threshold is also provided in Section \ref{stationary_state_analysis}.

The evaluation of the detector is based on when the first detection occurs, from here on denoted the time of detection, formally
\begin{equation}
    \tau = \{ \min t : \detector(t) = 1 \}.
\end{equation}
It is well-established that there is a trade-off between accuracy and false alarm in detection problems \cite{veeravalli2014, basseville1993}, 
and thus such a trade-off is also expected in this particular case. However, in time-series detection the standard accuracy measure is not well defined, and therefore the alternative performance metric of average detection delay ($\add$) \cite{veeravalli2014} is considered. Intuitively, this is the amount of time that detection takes, on average, after the event of interest occurs. In the current application, this is the time between a delay injection and the time of detection.
\begin{definition}
\label{def:add}
    The average detection delay ($\add$) is
    \begin{equation}
        \add(\tau) = \expect[\max(0, \tau - \attacktime)],
    \end{equation}
    where $\tau =  \{ \min t : \detector(t) = 1 \}$ and $\attacktime$ is the time of delay injection.
\end{definition}

A false alarm on time series data is a detection before the event of interest occur, and with this notion the probability of false alarm ($\pfa$) \cite{veeravalli2014} is defined as follows.

\begin{definition}
\label{def:pfa}
    The probability of false alarm ($\pfa$) is 
    \begin{equation}
        \pfa(\tau) = \prob(\tau < \attacktime),
    \end{equation}
    where $\tau =  \{ \min t : \detector(t) = 1 \}$ and $\attacktime$ is the time of delay injection.
\end{definition}

The problem can then be summarised as finding $\pi$ and $\xi$ in Definition \ref{def:detector} such that $\pfa$ and $\add$ are kept small enough for the specific problem at hand. 

\section{Delay detection by Estimation}
\label{delay_detection_by_estimation}
In this paper, the detection is performed by estimating the current delay in the system, and from this determining the \emph{posterior probability that there is no delay $(\delta_t=0)$} present. This probability constitutes the function $\pi$ in Definition \ref{def:detector}. Going forward, let $\pi^i_t \equiv \prob(\delta_t=i|\{y_k, u_{k-1}\}_{k=1}^t)$, 
to simplify notation. Then the proposed detector is
\begin{equation}
    \detector(t) = \begin{cases}0,&\pi^0_t > \xi \\
1, &\pi_t^0 \leq \xi \\
\end{cases}.
\label{eq:simple_alarm}
\end{equation}
To estimate the probability $\pi^0_t$ online, IMM filtering is applied. 

\subsection{IMM Filtering}
\label{interacting_multiple_models}
IMM filtering is an extension of Kalman filtering (for details on the Kalman Filter, see \cite{kailath2000}) to systems that switch between different \emph{modes} over time. A common application of IMM is air-target tracking, where the modes can represent turns, constant velocity flight, or hovering \cite{kirubarajan2003}. 
This was also the motivation for its development \cite{blom1984}. Inspired by the great success of IMM for air target tracking, this paper adapts IMM to work for delay detection. This is done by tracking the system state by considering switches between delay sizes $i$, with two main differences to standard IMM: (i) the state estimation is performed in a delayed manner to capture current delays, and (ii) the modes use the same state space model with different inputs, as opposed to standard IMM where the inputs are the same and state space models change between modes. The main idea of IMM, interaction between multiple models through estimate mixing, remains.

In many applications the aim of IMM is to accurately estimate the states, but in the current application the goal is rather to calculate the evolving mode posterior probabilities $\pi_t^i$, to be used in the detector. However, these probabilities and the state estimates are closely related in IMM filtering. In standard IMM filtering the posterior probabilities are calculated as a means to receive state estimates, whereas in the application to delay detection the state estimates are only calculated as needed to update the posterior probabilities. This distinction has no impact on the IMM filtering method in itself.

 The delay mode at time $t$ is here denoted as $\delta_t=i$, $i \in \mathbb{N}$. The transitions between modes are described by a Markov model with \emph{transition probabilities}
 \begin{equation}
     \prob(\delta_t=i|\delta_{t-1}=j).
     \label{eq:transition_probabilities}
 \end{equation}
 The conditional probability of the system being in mode $i$ at time $t$ is exactly the probability used by the detector \eqref{eq:simple_alarm}, and updated dynamically by the IMM filter.

 The notion of delay size as modes is believed to be novel and requires careful consideration of time indices, which goes beyond that of standard IMM. For a certain delay $\delta(t)=i$, the observation at time $t$ depends on the state at time $t-i$, \emph{not } on the state at time $t$. Thus, under $\delta_t=i$ time steps of delay, the state of the system is not observable until time $t+\delta_t$. To observe the full system, the observer needs to be delayed by the maximum considered delay $\maxdelay$. This is formalised in Proposition \ref{prop:delayed_detection}.

 \begin{proposition}
\label{prop:delayed_detection}
    Consider a system of the form \eqref{eq:state_space_model}, with $\delta_t=\maxdelay$. Then, the state at time $t$ is not observable before time $t+\maxdelay$.
\end{proposition}
\begin{proof}
    See appendix A.
\end{proof}

By Proposition \ref{prop:delayed_detection}, it is clear that standard IMM can not be directly applied to the problem. Motivated by this, the state estimation is delayed by $\maxdelay$ time steps, , so that at time $t$ the IMM filters estimate $x_{t-N}$. 

Another necessary divergence from conventional IMM can be found by studying the state space model in \eqref{eq:state_space_model}. 
Namely, due to the possible unknown delay, $y_t$ does not necessarily depend on $x_t$, see \eqref{eq:state_space_model}. This contradicts the assumptions for applying Kalman filters, a crucial part of the IMM filtering. 
The issue is resolved by considering a buffer of measurements $\{y_k\}_{k=t-\maxdelay}^t$, and feeding the appropriate measurement to each mode. \tcr{Specifically, the measurement $y_{t-N+i}$ is fed to the filter corresponding to mode $i$. Let $\etime \equiv t-\maxdelay$ denote the time for which state estimation is performed, }
and let $\tilde y_{\etime}$ \tcr{denote the observations that \emph{would} be obtained at time $\etime$ if no delay were present, i.e., $\tilde y_{\etime} = C x_{\etime}  + v_{\etime}$.}
Note that, under the hypothesis \tcr{assigned to mode $i$} that the delay is $\delta_t=i$,  $\tilde y_{\etime} = y_{\etime + i}$.

\begin{remark}
    Going forward, both $t$ and $\tilde t=t-\maxdelay$ will be used, to denote the actual time and the time of estimation respectively.
\end{remark}

With this reformulation \tcr{and routing of measurements from the buffer}, the standard state space model \tcr{for each mode} is recovered as
\begin{equation}
    \begin{cases}
        x_{\tilde t + 1} &= Ax_{\tilde t} +Bu_{\tilde t} +  w_{\tilde t} \\
        \tilde y_{\tilde t} &= Cx_{\tilde t} + \nu_{\tilde t} .
    \end{cases}
\end{equation}
By this, the difference between modes in the proposed version of IMM is not in the system dynamics $A,B,C$, but in the measurements used. Therefore, the stationary Kalman gain for all different modes is the same. 

The overall structure of the IMM filter can be seen in Fig. \ref{fig:IMM_diagram}. 
At each time step $t$, the observations $y_{\etime + i}$ and control actions \tcr{$u_{\etime}$} are given as input to the Kalman filter of mode $i$, and the output is an updated state estimate $\hat{x}_{\etime}^i$ for each potential delay $\delta_t=i$, as well as updated probabilities of the different modes, $\pi_t^i$.

The final IMM algorithm can now be described in detail, and the full description is provided as Algorithm \ref{alg:IMM}. At time $t$, the first step of IMM is to \emph{calculate mixing probabilities}. For each possible delay $\delta_t$, this is the probability of each previous potential delay $\delta_{t-1}$ given the previous estimates and the system dynamics, e.g. for  $\delta_t=i$ and $\delta_{t-1}=j$, the mixing probability is
\begin{equation}
   \prob(\delta_{t-1}=j|\delta_t=i, \mathcal{I}_{t-1}) = \frac{\pi_{t-1}^j\prob(\delta_t=i|\delta_{t-1}=j)}{\prob(\delta_t=i)},
\end{equation}
where $\mathcal{I}_t$ denotes the available information \tcr{$\{u_{0:t-1},y_{1:t}\}$}; $\prob(\delta_t=i|\delta_{t-1}=j)$ is the Markov model transition probabilities, taking on the role of prior; and $\pi^j_{t-1}$ is given as output of the IMM filter at time $t-1$, thus implicitly depending on $\mathcal{I}_{t-1}$, and taking the role of likelihood. 
Intuitively, the mixing probabilities describe the \emph{probability of the system being in a certain mode at the previous time step, given a hypothesis on the current mode}. The current state is of course not known, so these probabilities are calculated for all possible hypotheses on $\delta_t$.

The next step is to use these mixing probabilities to weigh the previous estimates. Under the hypothesis $\delta_t=i$, an estimate on the previous state of the system is given as
\begin{equation}
    \tilde x_{\etime-1|\etime-1}^i = \sum_{j=0}^\maxdelay\prob(\delta_{t-1}=j|\delta_t=i)\hat x_{\etime-1|\etime-1}^j,
\end{equation}
where $\hat x_{\etime-1|\etime-1}^j$ denotes the estimate of the state at time $\etime-1$ given that the system is in mode $j$ at time $t-1$, and $\tilde x_{\etime-1|\etime-1}^i$ denotes the estimate of the state at time $\etime-1$ given the assumption that the system is in mode $i$ at time $t$. This definition of mixing is a key to the high performance observed in air target tracking applications \cite{bar2001}. 

Third, a Kalman filter is applied to each of the mixed estimates $\{\tilde x_{\etime-1|\etime-1}^i\}_{i=0}^\maxdelay$, 
\tcr{using the inputs $y_{\etime+i}$ and $u_{\etime-1}$,}
to return new estimates $\{\hat x_{\etime|\etime}^i\}_{i=0}^\maxdelay$. 

Independently to the state estimate update, the mode probabilities are updated based on the new data as
\begin{equation}
    \pi_t^i = \frac{\prob(y_{\etime}|\delta_t=i, \tilde x_{\etime-1|\etime-1}^i)\prob(\delta_t=i| \pi_{t-1}^i)}{\prob(y_{\etime})},
\end{equation}
where 
\begin{equation}
\label{eq:IMM_step3a_eq2}
    \prob(\delta_t=i| \pi_{t-1}^i) = \sum_{j=1}^\maxdelay \prob(\delta_t=i|\delta_{t-1}=j) \pi_{t-1}^j,
\end{equation}
and $\prob(y_{\etime}|\delta_t=i, \tilde x_{\etime-1|\etime-1}^i)$ calculated from observations.

\begin{figure}
    \centering
     \tikzset{
block/.style = {align=center, draw, fill=white, rectangle, minimum height=3em, minimum width=3em, align=center},
dottedblock/.style = {align=center, draw=black, fill=white, dotted, rectangle, minimum height=3em, minimum width=3em, align=center},
blueblock/.style = {align=center, draw=blue, fill=white, rectangle, minimum height=3em, minimum width=3em, align=center},
tmp/.style  = {coordinate}, 
sum/.style= {draw, fill=white, circle, node distance=1cm},
input/.style = {coordinate},
output/.style= {coordinate},
pinstyle/.style = {pin edge={to-,thin,black}
}
}
\scalebox{0.8}{
\begin{tikzpicture}[auto, node distance=3cm,->]

\node [block, name=mixing_probability,] (mixing_probability) {Step 1: \\ Calculate \\ mixing probabilities};

\node [block, name=mixing, below of=mixing_probability] (mixing) {Step 2: \\ Mix \\previous estimates};

\node [block, name=KF, below of=mixing] (KF) {Step 3b: \\ Kalman filter \\ estimation};

\node[block, name=mode_prob_calc, below of = KF](mode_prob_calc){Step 3a: \\ Update \\ mode probabilities};

\node[blueblock, name=detector, below of=mode_prob_calc](detector){Detector: \\ $\detector(\pi_t^0)$};

\node[input, name=observations1, left of=KF](observations1){};
\node[input, name=observations2, left of=mode_prob_calc](observations2){};

\node[output, name=output_est, right of=KF, node distance=6cm](output_est){};
\node[output, name=output_prob, right of=mode_prob_calc, node distance=6cm]{};

\node[input, name=corner_a1, right of = KF, node distance = 3cm](corner_a1){};
\node[block, name=corner_a2, below of = corner_a1, node distance=1.5cm](corner_a2){$q^{-1}$};
\node[input, name=corner_a3, above of=mode_prob_calc, node distance=1.5cm](corner_a3){};

\node[input, name=corner_b1, right of=KF, node distance=4cm](corner_b1){};
\node[block, name=corner_b2, right of=mixing, node distance=4cm](corner_b2){$q^{-1}$};

\node[input, name=corner_c1, right of=mode_prob_calc, node distance=5cm](corner_c1){};
\node[input, name=corner_c2, above of = corner_c1, node distance = 2.7cm](corner_c2){};
\node[input, name=corner_c3, above of=corner_c2, node distance =0.6cm](corner_c3){};
\node[input, name=corner_c4, right of=mixing_probability, node distance=5cm](corner_c4){};

\node[output, name=output_detector, right of=detector, node distance=6cm](output_detector){};

\draw[->] (mixing_probability) -- node{$\{\prob(\delta_{t-1}=j|\delta_t=i)\}_{i, j = 0}^\maxdelay$}(mixing);

\draw[->] (mixing) -- node{$\{\tilde{x}_{\tilde t-1|\tilde t -1}^{j}, P_{\tilde t-1| \tilde t_1}^{j}\}_{j=0}^{\maxdelay}$}(KF);

\draw[->] (mode_prob_calc) -- node{$\{\pi_t^i\}_{i=0}^{\maxdelay}$}(detector);

\draw[->] (observations1) -- node{$\{y_{\etime+i}\}, u_{\etime-1}$}(KF);
\draw[->] (observations2) -- node{$\{y_{\etime+i}\}, u_{\etime-1}$}(mode_prob_calc);

\draw[-] (KF) -- node{$\{ \hat{x}_{\tilde t}^i, P_{\tilde t}^i\}_{i=0}^{\maxdelay}$}(corner_b1);
\draw[->] (corner_b1) -- (output_est);

\draw[->] (mode_prob_calc) -- node{$\{\pi_t^i\}_{i=0}^{\maxdelay}$}(output_prob);

\draw[->] (corner_a1) -- (corner_a2);
\draw[-] (corner_a2) -- (corner_a3);
\draw[->] (corner_a3) -- (mode_prob_calc);

\draw[->] (corner_b1) -- (corner_b2);
\draw[->] (corner_b2) -- (mixing);

\draw[-] (corner_c1) --(corner_c2);
\draw[-] (corner_c3) -- (corner_c4);
\draw[->] (corner_c4) --(mixing_probability);
\draw[-] (5,-6.3) arc[start angle=-90, end angle=90, radius=0.3cm];

\draw[->] (detector) -- node{$0$ or $1$}(output_detector);

\end{tikzpicture}}
    \caption{Block diagram over IMM for delay detection. The blue box is the actual detection step. The operator $q^{-1}$ denotes a unit delay shift operator. }
    \label{fig:IMM_diagram}
\end{figure}
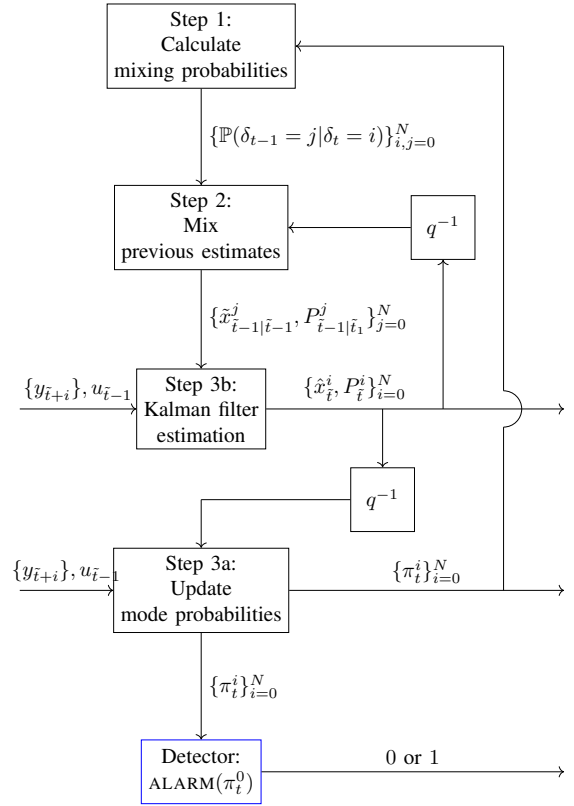

\begin{algorithm}
\caption{IMM}\label{alg:IMM}
\begin{algorithmic}
\State \textbf{Input:} $(y_k, u_{k-1})_{k=\etime}^t$ \\
\State \textbf{Output:} $(\pi^i_t, \hat x_{\etime}^i)_{i=0}^{\maxdelay-1}$ \\
\State  \textbf{Step 1:} The mixing probabilities are calculated using Bayes' rule as 
\begin{equation}
\label{eq:IMM_step1}
    \prob(\delta_{t-1}|\delta_t, \mathcal{I}_{t-1}) = \frac{\prob(\delta_t|\delta_{t-1})\pi_{t-1}}{\prob(\delta_t)},
\end{equation}
where $\prob(\delta_t|\delta_{t-1})$ are the transition probabilities, and the denominator can be treated as a normalization. \\
\State \textbf{Step 2: } The previous estimates and covariances are mixed as
\begin{equation}
\label{eq:IMM_step2}
    \tilde{x}_{\etime-1|\etime-1}^{i} = \sum_{\text{j}=1}^\maxdelay \hat{x}_{\etime-1|\etime-1} \prob(\delta_{t-1} = j|\delta_t = i),
\end{equation}
and 
\begin{equation}
\label{eq:IMM_step2_eq2}
    P_{\etime-1}^{i} = \sum_{j=1}^\maxdelay P_{\etime-1}^j \prob(\delta_{t-1} = j|\delta_t = i),
\end{equation}
for each mode $i$, using the mixing probabilities from \eqref{eq:IMM_step1}. \\
\State \textbf{Step 3a:} The new mode probabilities are calculated as
\begin{equation}
\label{eq:IMM_step3a}
    \pi_t^i = \frac{\prob(y_{\etime}|\delta_t=i, \tilde x_{\etime-1|\etime-1}^i)\prob(\delta_t=i)}{\prob(y_{\etime})},
\end{equation}
where 
\begin{equation}
\label{eq:IMM_step3a_eq2}
    \prob(\delta_t=i) = \sum_{j=1}^\maxdelay \prob(\delta_t=i|\delta_{t-1}=j) \pi_{t-1}^j,
\end{equation}
and $\prob(y_{\etime}|\delta_t=i, \tilde x_{\etime-1|\etime-1}^i)$ is calculated from observations. \\
\State \textbf{Step 3b: } The new estimates are calculated, for each $i = 0, 1, 2, ... \maxdelay$, by applying the Kalman filter to the mixed estimate distributions $\{\tilde{x}_{\etime-1}^i, P_{\etime-1}^i\}_{i=0}^\maxdelay$ from step 2 and the observation matching the mode hypothesis, $y_{\tilde t +i}$. 

\end{algorithmic}
\end{algorithm}

\subsection{Formulation of a Parsimonious Markov Model}
\label{construction_of_the_markov_model}

 To apply IMM to delay detection, a Markov model describing the potential transitions between delay modes is needed. This describes both which transitions are possible, and how probable they are. Known or assumed information about this behaviour can be encoded in the transition probabilities. For example, there could be physical limitations or other safeguards that restrict the ways delay can be injected, and this information could be encoded in the Markov Model by changing probabilities or disabling some transitions.  

 However, when there are no physical limitations to the attack behaviour, a parsimonious model is more robust. Even if there were historical patterns on how delay attacks are implemented that could be encoded in the Markov model, detection methods optimised for these would encourage deviations from these patterns, thus invalidating the imposed assumptions. Therefore, this paper follows Assumption \ref{ass:delay} and does not incorporate any detailed information into the Markov model.




When designing the Markov model, it is important to note that without any imposed structure, $\maxdelay(\maxdelay-1)$ parameters need to be set when considering $\maxdelay$ potential delays, by the law of total probability. To allow for informed parameter tuning, a smaller set of parameters is needed. Thus, a general Markov model is considered where the system without observations changes between delay modes $\delta_t>0$ with uniform probability. This should not be viewed as a narrow assumption on system or attack behaviour, but rather as the most general prior that can be put on the transition probabilities, allowing the observed data to steer the estimation as much as possible. By the same reasoning, and standard IMM practice \cite{bar2001}, transition probabilities are considered constant over time and independent of observations. With these simplifications, the only transition probabilities that need to be set are
 \begin{equation}
 \label{eq:p0}
     p_0 \equiv \prob(\delta_t>0|\delta_{t-1}=0) = \sum_{j=1}^{\maxdelay-1}\prob(\delta_t=j|\delta_{t-1}=0),
 \end{equation}
 and 
 \begin{equation}
 \label{eq:p1}
     p_1 \equiv \prob(\delta_t=0|\delta_{t-1}>0) =\sum_{j=1}^{\maxdelay-1} \prob(\delta_t=0|\delta_{t-1}=j),
 \end{equation}
 where $p_0$ is the probability of transitioning from zero delay to some delay, and $p_1$ the probability to transition from some delay to no delay. By Assumption \ref{ass:delay}, it is tempting to set $p_1=0$ to capture the knowledge that 
the delay injection attack is persistent. However, this is detrimental to the performance of the method as it leads to asymptotic detection when no attack is present, without observations. This behaviour of the IMM detector can be understood by analysing the stationary distribution of the Markov Model. 

\subsection{Setting Parameters via the Stationary Distribution}
\label{stationary_state_analysis}

Independent of the initialisation, a Markov model always converges to some distribution. This distribution is the \emph{stationary distribution} of the model.

The stationary distribution is a valuable tool in the analysis of the IMM detection scheme. As the aim is only to differentiate between there being a delay or not, the interesting part is the stationary distribution over these two (groups of) modes, see \eqref{eq:simple_markov_model}. This is fully determined by the stationary probability that no attack is present, from here on denoted $\statprob^0$. 


To analyse the stationary distribution, consider the simplified Markov model resulting from \eqref{eq:p0} and \eqref{eq:p1}, whose transition matrix can be represented as 
\begin{equation}
\label{eq:simple_markov_model}
    \begin{bmatrix}
        1-p_0 & p_1 \\
        p_0 & 1-p_1
    \end{bmatrix}.
\end{equation}
As an example of how the stationary distribution can provide new insights, such analysis shows why setting $p_1=0$ is always a bad idea, no matter if this is the true nature of an attack. Specifically, this results in a stationary distribution of certain delay, which means no choice of $\xi$ can save the detector from guaranteed false positives. This is formalised as Proposition \ref{prop:p1=0}.

\begin{proposition}
\label{prop:p1=0}
    Consider the Markov model \eqref{eq:simple_markov_model}. If $p_0>0$ and $p_1=0$, the stationary distribution is $\statprob^0=0$.
\end{proposition}
\begin{proof}
It is easily verified that
\begin{equation}
    \begin{bmatrix}
        1-p_0 & 0 \\
        p_0 & 1
    \end{bmatrix} \begin{bmatrix}
        0 \\ 1
    \end{bmatrix} = \begin{bmatrix}
        0 \\1
    \end{bmatrix}, 
    \end{equation}
    and by definition of a stationary distribution this directly gives $\statprob^0=0$.
\end{proof}

The same principles can be used as a tool to design the Markov model parameters. 
The case of uninformative (or equivalently, missing) observations is considered. Then, the stationary distribution gives the probability of nominal mode $\statprob^0$ by:
\begin{equation}
\label{eq:stationary_dist_tosolve}
\begin{bmatrix}
    1 -p_{0} & p_{1} \\
    p_{0} & 1-p_{1}
\end{bmatrix} \begin{bmatrix}
    \statprob^0\\
    1-\statprob^0
\end{bmatrix} = \begin{bmatrix}
    \statprob^0 \\
    1-\statprob^0
\end{bmatrix},
\end{equation}
and solving for $p_{0}$ yields:
\begin{equation}
\label{eq:p1_reparametrised}
    p_{0} = p_1 \frac{1-\statprob^0}{\statprob^0} .
\end{equation}
Using this relation between $p_0$ and $p_1$, the full Markov model can be determined by a desired stationary distribution and some domain knowledge of $p_0$. However, this does not guarantee a $p_1$ that is a probability, i.e. $\leq 1$, and $p_0$ may not be known. Instead, it can be noted that if $\statprob^0 \geq 0.5$, which is a reasonable design choice in most applications, $\frac{1-\statprob^0}{\statprob^0} \leq 1$ and thus $p_1$ can be chosen freely in $[0, 1]$. Further, this allows interpreting $p_1$ as the sensitivity of the model to new observations. A low $p_1$, and thus a low $p_0$, implies slow transitions between modes and thus more observations are needed to change the estimated probabilities. By instead using a $p_1$ close to $1$, the system adapts quickly and $\add$ is minimised. 

%

The specification of transition probabilities in \eqref{eq:simple_markov_model} is now parametrized by the stationary probability $\statprob^0 \in (0,1)$ and the tuning parameter $p_1 \in (0, 1)$. 


By using this simplified Markov Model and the above reparametrisation, \emph{the full Markov model of the IMM filter is parametrised by two variables, no matter how many delay modes are considered.} Without this simplification, $\maxdelay(\maxdelay-1)$ parameters would be needed to describe the model, all of which are hard to get from domain knowledge.

The choice of detection threshold of detectors is an open question in general, and this paper does not claim to solve it. Some guiding principles for setting the threshold is that the stationary distribution without attack, i.e. $\statprob_0$, should \emph{not} trigger detection, while the stationary distribution under attack \emph{should} yield detection. The latter can be estimated by simulation, and in the experiments of this paper this is estimated as the stationary distribution under an average delay attack, according to the Markov model. Tuning can also be done directly by visualisation of simulated data. More sophisticated methods to pick the threshold is left as future work.

\section{DELAY DETECTION IN CRUISE CONTROL}
\label{delay_detection_in_cruise_control}

\subsection{Cruise Control Model}
\label{cruise_control_model}
 A cruise control system is an engineering example of a dynamic system that is switched between open and closed loop settings.
Due to the delay being injected in the feedback loop, the attack is in general not noticeable until the loop is closed. However, when the loop is closed, the attack is effective and may destabilise the system; see \cite{wigren2023}. The importance of the example is motivated by the increased usage of cruise control systems, and the potentially disastrous outcomes of destabilisation. Therefore, such attacks need to be detected before the loop is closed.
A linearised state space model for the car dynamics is
\begin{equation}
\label{eq:cardynamics}
    \begin{cases}
        x_t = \Delta x_{t-1} + u_{t-1} + w_t \\
        y_t = x_t + \nu_t
    \end{cases},
\end{equation}
where $\Delta$ describes the expected amount of speed lost in one time step. This model can be obtained from a non-linear physics-based model. Such a derivation can be found in \cite{wigren2023}. 

In all experiments, to mimic manual control  by a human driver, a simple proportional controller is used, which directly feeds the error signal, without delay, as its control action, with an imposed saturation of $|u| < 5$.

The attack is implemented by storing both the true observations (the observations that would have been seen if there where no attack), and the actual observations. These are equal until the time of attack, but when a delay of $\delta_t$ is injected the actual observation at time $t$ instead equals the true observation at time $t-\delta_t$, similar to the previous intuition from Fig. \ref{fig:example_delay}. Depending on the experiment, the attack time $\attacktime$ is either fixed or uniformly random on the time scale used.

As model parameters, $\Delta = 0.99$ and independent $w_t \sim \mathcal{N}(0, 1)$ and $\nu_t \sim \mathcal{N}(0, 0.01)$ are used. The variance of the observations needs to be lower than that of the model to obtain good performance, since detection based on observations can only happen if the observations are significantly more trustworthy. This is a realistic assumption in the cruise control example, as the velocity state can be measured with the Global Positioning System (GPS) with high accuracy, while e.g. steep hilly terrain may significantly deviate from the general car model. Uncertainty in the model specification is also captured by the system noise.

The system works over time increments corresponding to $0.1$ seconds, and there are IMM filters for delays of $0, 1, 2, ...,10$ time steps, i.e. between $0.0$ and $1.0$ seconds, where the $0.0$ second delay filter corresponds to the nominal mode. 

The initial estimates are the true state, and the simulations are started with a speed of $60$km/h and a reference speed of $60$km/h.
The Kalman filters are designed using the stationary Kalman gain \cite{kailath2000}. 
The $\add$ and $\pfa$ are estimated from simulations by Monte Carlo estimates of the probabilities and expectations, with some fixed probability distribution of the attack point $\attacktime$.

\subsection{Tuning of Hyper Parameters}
\label{tuning}

In this paper, a strong prior belief of $90\%$ chance of the nominal mode and equal chance for the different delays is used as initialisation. The transition probabilities are chosen to guarantee that the stationary distribution with no (or equivalently, uninformative) observations is the same as the initial prior (i.e., $\prob(\delta_0 = 0) = \pi^0 = 0.9$ ), so that any change in the distribution can be interpreted as an actual change in probability based on the observations.

\begin{figure}
    \centering
    \includegraphics[width=0.9\linewidth]{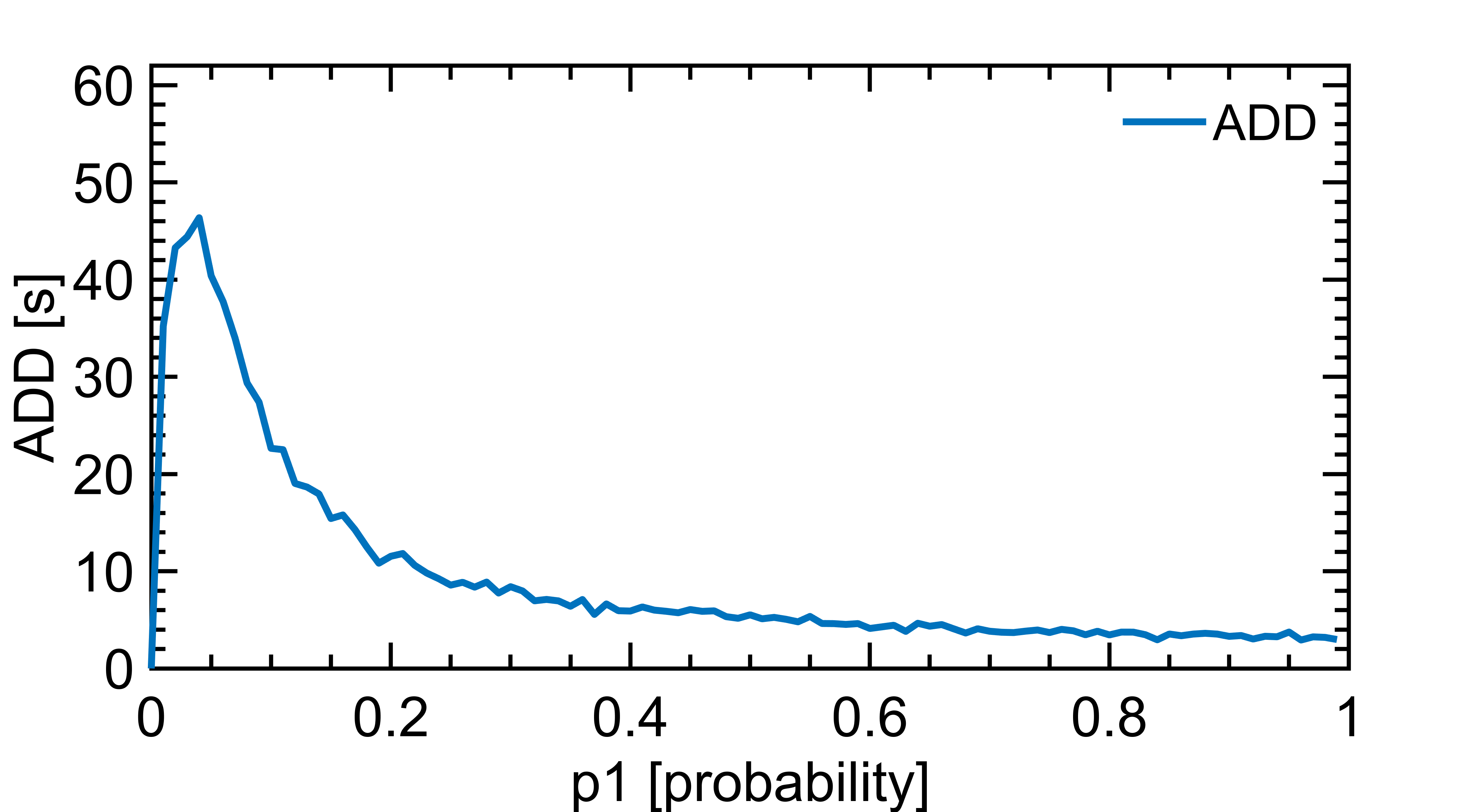}
    \caption{$\add$ depending on the tuning parameter $p_1$. The plot shows the average over $50$ runs, with the detection threshold fixed at $0.4$, and the delay of $0.2$ seconds always being added after $250$ seconds.}
    \label{fig:heattuning_ADD}
\end{figure}
\begin{figure}
    \centering
    \includegraphics[width=0.9\linewidth]{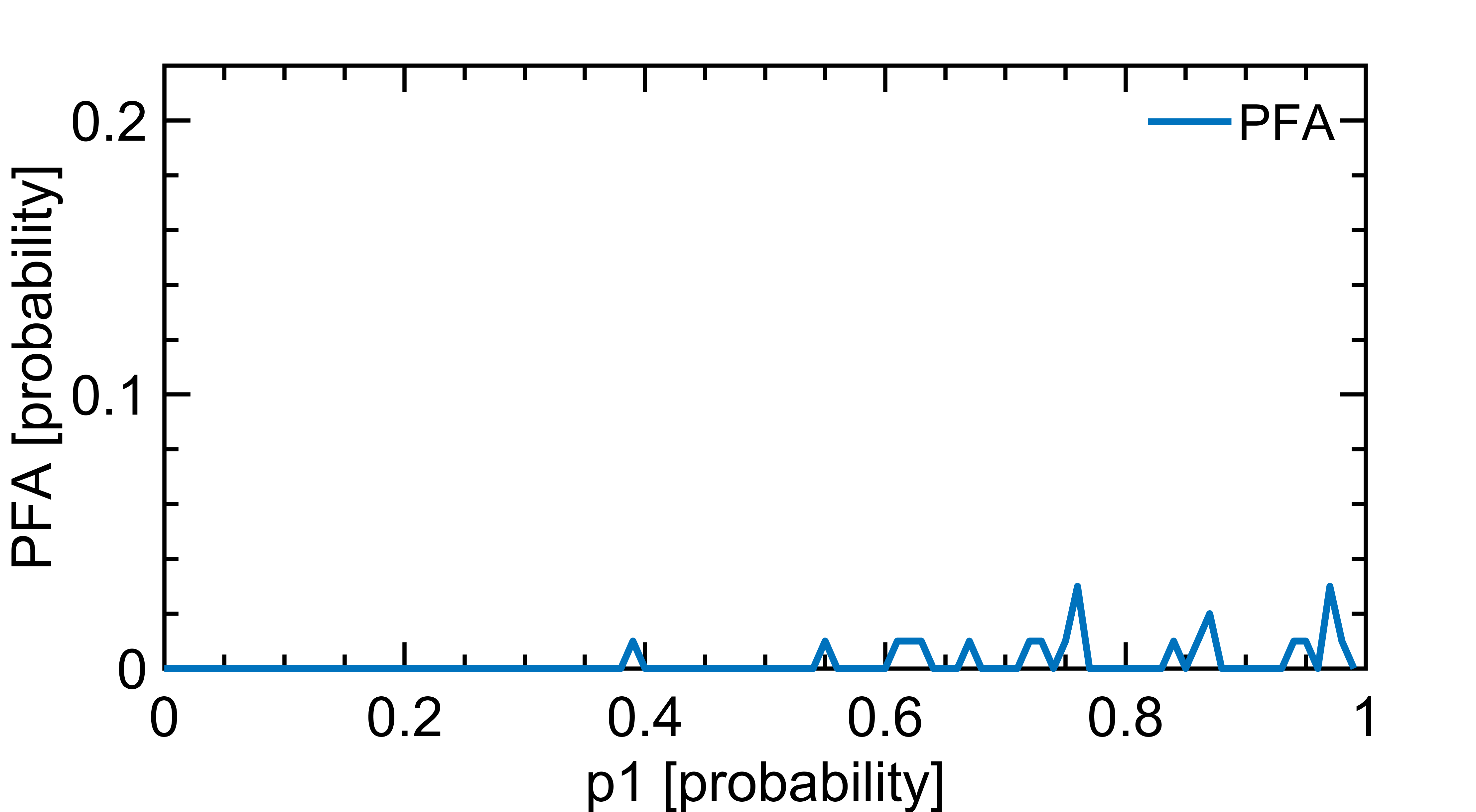}
    \caption{$\pfa$ depending on $p_1$. The plot shows the average over $50$ runs, with the detection threshold fixed at $0.4$, and the delay of $0.2$ seconds always being added after $250$ seconds.}
    \label{fig:heattuning_PFA}
\end{figure}

By applying \eqref{eq:p0}  to the desired stationary distribution ($90\%$ probability of nominal mode), the transition $p_0$ is given by
\begin{equation}
\label{eq:transition_parameters}
        p_0 = p_1 \frac{1-\statprob^0}{\statprob^0} = p_1\frac{0.1}{0.9} \approx 0.11 p_1
\end{equation}
with $0 \leq p_1 \leq 1$. The tuning of $p_1$ is done experimentally. The threshold is fixed to $\xi=0.4$, and the attack is always applied after $250$ seconds, introducing a delay of $0.2$ seconds. In Figs. \ref{fig:heattuning_ADD} and \ref{fig:heattuning_PFA}, the $\add$ and $\pfa$ are plotted as functions of $p_1$, which is varied between $0.0$ and $0.99$. The metrics are calculated by averaging over $50$ runs for each value of $p_1$. It can be seen that choosing $p_1$ is strongly related to finding a suitable trade-off between false alarm and detection delay. However, with the current threshold $\pfa$ is less than $0.5\%$ for all $p_1$. Thus, for the rest of the experiments a high $p_1=0.99$ is implemented to minimise the $\add$, although it is noted that a smaller $p_1$ is useful in cases where one would like to further decrease the $\pfa$.


Experimentally, the mean likelihood after an attack is introduced is estimated over $50$ runs to be $0.20$, with a $99$th quantile of $0.28$. Using this to calculate the stationary distribution under attack and the intuition that the threshold $\xi$ should be no smaller than this, it is concluded that the detection threshold should be in the interval $(0.28, 0.9)$.
Since the probabilities vary around these stationary distributions, some additional safety-margin should be included. Prioritising a lower false alarm probability, a threshold of $0.4$ is chosen.

\begin{figure}
    \centering
    \includegraphics[width=0.9\linewidth]{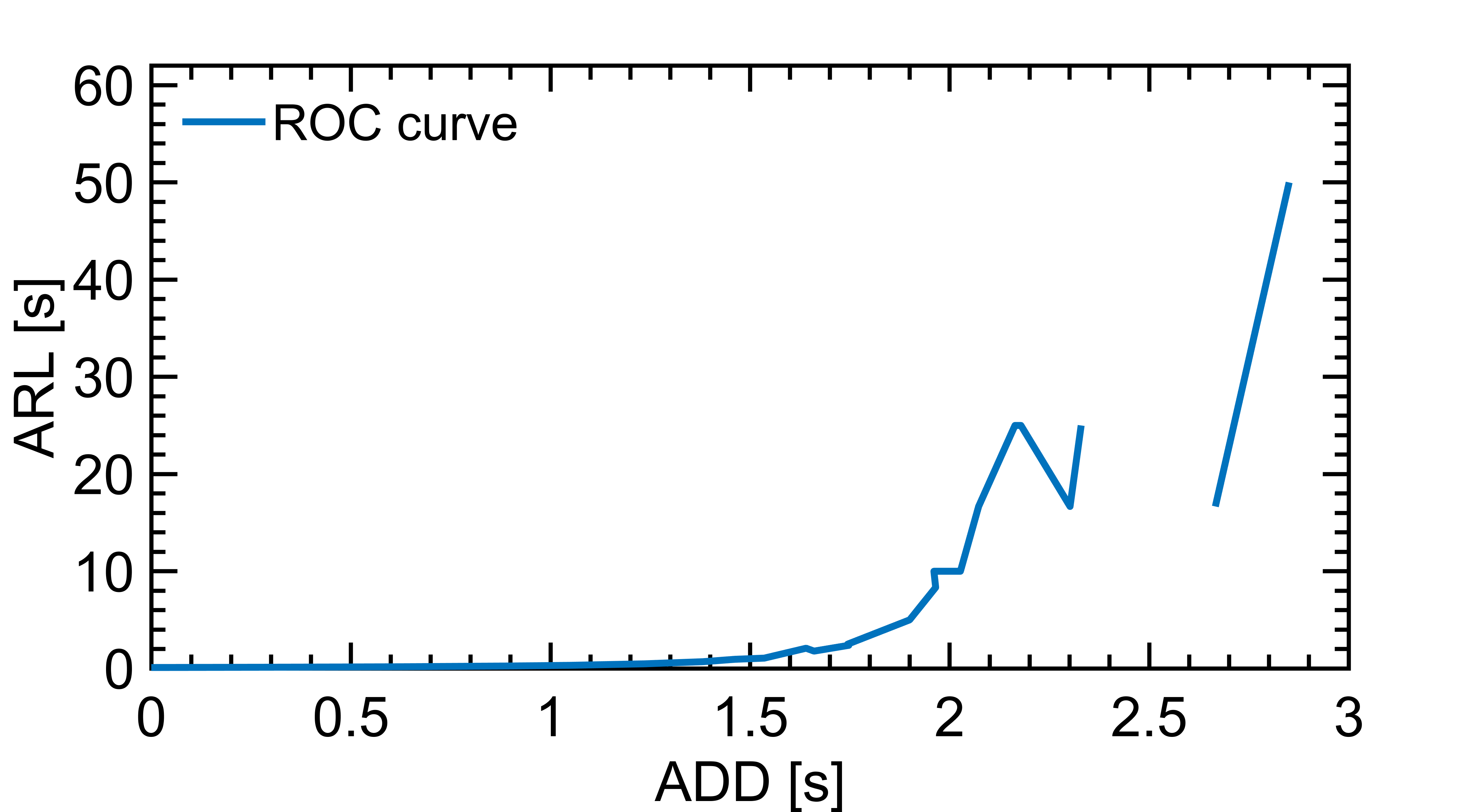}
    \caption{Empirical ROC curve of $\add$ and $\arl \equiv 1/\pfa$ in seconds, with the threshold for detection varied between $0.01$ and $0.99$. The delay is of $0.2$ seconds and introduced after $20$s.}
    \label{fig:ROC_curve}
\end{figure}

In Fig. \ref{fig:ROC_curve}, an empirical ROC curve for the detector is shown. As suspected, a low average detection delay corresponds to a lower average run length ($\arl \equiv \frac{1}{\pfa} \cdot \frac{\text{time steps}}{s}$), and the trade-off is relatively beneficial. 

\subsection{Simulation Results}
\label{simulation_results}
All simulation results consider a delay attack defined as a step function, with $\delta_t=0$ until some attack time $\attacktime$, when a fixed delay is inserted.

In Figure \ref{fig:single_trajectory}, a single trajectory of the IMM filter estimate is shown. It is clear from the figure that the filter is quick to adjust the probabilities when the attack is introduced. A threshold of $0.4$ is supported by this trajectory as well, as it is clear that it quickly detects the attack without any false alarms, in this instance. 

\begin{figure}
    \centering
    \includegraphics[width=0.9\linewidth]{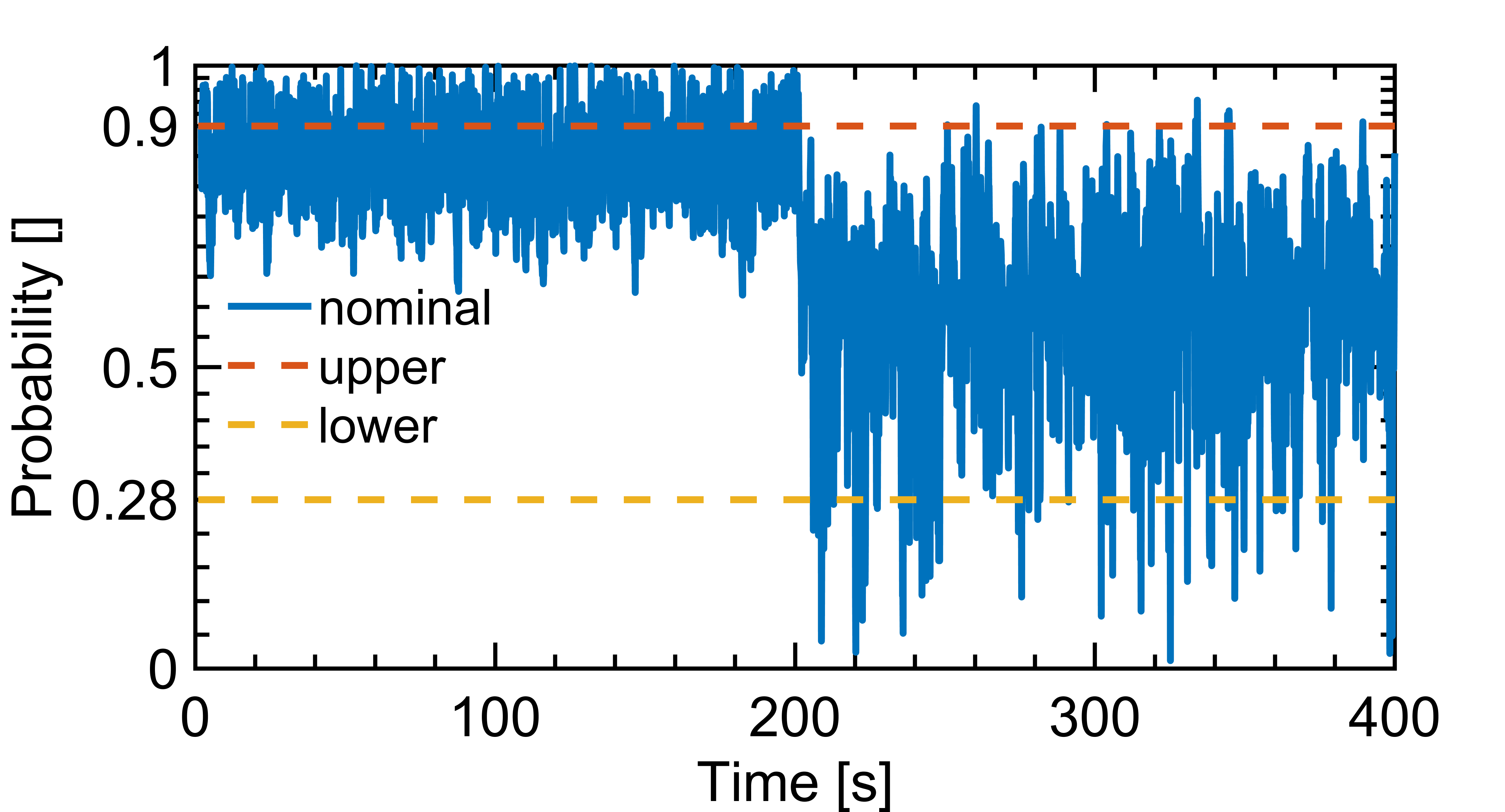}
    \caption{A single trajectory of how the nominal probability change over time, using $p_1=0.99$ and a delay attack of $0.2$ seconds being introduced after $200$ seconds. The upper and lower bound for reasonable threshold from the estimated range $(0.28, 0.9)$
    are shown for reference.}
    \label{fig:single_trajectory}
\end{figure}

In Table \ref{table:various_delays}, the estimated $\add$ and $\pfa$ are shown after running the model for $5000$ runs for injected delays of $0.1, 0.2, ..., 1.2$ seconds. The attack time is sampled uniformly random between $2$ and $302$ seconds. It should be noted that the two largest delays in this simulation, $1.1$ and $1.2$ seconds, are higher than any of the hypotheses used by the model, and the detector performs slightly worse for these. The performance is still good though, which suggests that detection is possible even for attack sizes that were not anticipated. The main outlier is the $0.1$s delay, which takes on average over $90$ seconds to detect. 
It is hypothesised that this delay has a very minor effect on the observations, which makes it hard to detect. Understanding this better and improving the method for small delays is left as future work.
Provided that the cruise controller is designed to be robust against such low delays, for example as in \cite{wigren2023}, a high detection time for low delays is not a problem in practical applications. 


It can be noted that the implemented attack behaviour  (i.e., as a persistent delay injected from time $t_{*}$) differ significantly from the implemented Markov model. Therefore, the good performance further supports the usage of a robust parsimonious Markov model.

\begin{table}[h]
\caption{Estimated $\add$ and $\pfa$ over $500$ runs, for varied delays attacks inserted at a uniformly random time. }
\begin{center}
\begin{tabular}{|c||c|c|}
\hline
Delay [s] & $\add$ [s] & $\pfa$ [probability]\\
\hline\hline
$0.1$ &  $92.9$ & $0.0055$ \\
$0.2$ & $3.16 $ & $0.0052 $ \\
$0.3$ & $3.13 $ & $0.0052$ \\
$0.4$ & $3.17$ & $0.0050$ \\
$0.5$ & $3.15$ & $0.0040$ \\
$0.6$ & $3.21 $ & $0.0046 $ \\
$0.7$ & $3.15$ & $0.0052$ \\
$0.8$ & $3.15$ & $0.0054$ \\
$0.9$ & $3.12$ & $0.0052$ \\
$1.0$ &  $3.15$ & $0.0072$ \\
$1.1$ & $3.67$ & $0.0048$ \\
$1.2$ &  $4.37$ & $0.0060$\\
\hline
\hline
Average & $10.78$ & $0.0053$\\
Average excl. $0.1$ & $3.31$ & $0.0053$\\
\hline
\end{tabular}
\end{center}
\label{table:various_delays}
\end{table}

\subsection{Comparison with Existing Detection Schemes}
\label{comparison_with_existing_detection_schemes}
Previous work on delay detection builds on identification of the system parameters. In \cite{wigren2023}, an identification scheme is applied to the same vehicle model as in this paper, but assuming unknown system dynamics. This gave a detection time exceeding $100$ seconds on an example trajectory, which is significantly worse than the results for the IMM framework. They are not fully comparable, since the IMM approach assumes known model dynamics, but it shows that much better performance is possible when the dynamics are known. With small uncertainties in the model, the IMM model is believed to be able to handle them as long as the threshold is chosen with this uncertainty in mind. 

It is harder to compare the current work to that of \cite{korkmaz2017}, since those experiments are not performed on the same application and do not use a fixed-threshold detector. However, using that method, it is possible to note a difference in the distribution after around $10$ seconds, slightly worse than the IMM method, but only if the delay is inserted momentarily. A gradually deployed attack is not noticeable since the method assumes a nominal mode for the previous time step so any uncertainty of previous modes are discarded. This issue should not be present in IMM, since the uncertainties of the mode from past time steps are carried over to future steps.

\section{CONCLUSIONS}
\label{conclusions}
The paper described how Interacting Multiple Model filtering can be used to detect delay attacks in the feedback path. By derivation of stationary distributions, both the transition probabilities and threshold could be tuned in an informed manner. The tuned model detected delay attacks in a few seconds on a simulation of a cruise control application. Interesting directions for future work is to find theoretical bounds on the false alarm probability given the threshold, or to detect delay attacks in closed loop settings using the IMM approach.

\addtolength{\textheight}{-12cm}   



\section*{APPENDIX}
\subsection{Proof of Proposition 1}
\begin{proof}
    Let $$\textbf{x}_{t} = \begin{bmatrix}
        x_t & x_{t-1} & ... x_{t-\maxdelay}
    \end{bmatrix}^T.$$ Using this, \eqref{eq:state_space_model} can be expressed as a proper state space model with $y_t$ depending on $\textbf{x}_t$, with new system dynamics $$\textbf{A} = \begin{bmatrix}
        A & 0 & ... &0 & 0 \\
        I & 0 & ... & 0 &0 \\
        & . \\
        & & . \\
        & & & . \\
        0 & 0 & ... & I  & 0 \\
    \end{bmatrix} \text{ and } \textbf{C} = \begin{bmatrix}
        0 & ...& C
    \end{bmatrix}.$$ 
    Then, the observability matrix is
    $$\begin{bmatrix}
        0 & ... & 0 & C \\
        0 & ... & 0 & 0 \\
        & \vdots \\
    \end{bmatrix},$$
    so only state $x_{t-\maxdelay}$ is observable.
\end{proof}





\end{document}